\documentclass[10pt,journal,final,twocolumn]{IEEEtran}

\ifCLASSINFOpdf
\else
   \usepackage[dvips]{graphicx}
\fi
\usepackage{graphicx}
\usepackage{amsmath,bm}
\usepackage{amsmath} 
\usepackage{amsfonts,amssymb} 
\usepackage{extarrows}
\def\bh{{\bf h}}

\def\bI{{\bf I}}
\def\Exp{{\mathbb E}}
\def\b0{{\bf 0}}
\newtheorem{theorem}{Theorem}
\newtheorem{corollary}{Corollary}

\newenvironment{proof}{{\it Proof}\quad}{\hfill \par}
\usepackage{cite}
\usepackage{color}

\begin{document}

\title{NOMA Design with Power-Outage Tradeoff for Two-User Systems}

\author{Zeyu Sun, Yindi Jing \IEEEmembership{Member, IEEE} and Xinwei Yu
\thanks{The authors are with the University of Alberta, Edmonton, Alberta, Canada T6G2R3.  (email: zsun8@ualberta.ca; yindi@ualberta.ca; xinwei2@ualberta.ca)}
}

\maketitle

\begin{abstract} 
This letter proposes a modified non-orthogonal multiple-access (NOMA) scheme for systems with a multi-antenna base station (BS) and two single-antenna users, where NOMA transmissions are conducted only when the absolute correlation coefficient (CC) between the user channels exceeds a threshold and the BS uses matched-filter (MF) precoding along the user with the stronger average channel gain.  We derive the average minimal transmit power to guarantee the signal-to-interference-plus-noise-ratio (SINR) levels of both users. Our results show that the average minimal power grows logarithmically in the reciprocal of the CC threshold and a non-zero threshold is necessary for the modified NOMA scheme to have finite average minimal transmit power. Further, for the massive MIMO scenario, we derive the scaling laws of the average transmit power and outage probability with respect to the antenna numbers, as well as their tradeoff law.  Simulation results are shown to validate our theoretical results.
\end{abstract}


\IEEEpeerreviewmaketitle

\section{Introduction}

\IEEEPARstart{N}{on-orthogonal} multiple-access (NOMA) is a potential technique for the next generation mobile communications. With superposition coding at the transmitting base station (BS) and successive interference cancellation (SIC) at the receiving users\cite{huang}, NOMA can provide improved spectral efficiency via the extra degree-of-freedom in the power domain compared to traditional orthogonal multiple-access (OMA) schemes \cite{saito}.

Early works on NOMA considered the single-input-single-output (SISO) case, where the BS is equipped with a single-antenna, e.g., \cite{ding1,chen,ding2}. For NOMA with multiple uniformly distributed users, in \cite{ding2}, expressions were derived for both the sum-rate and the outage probability of each individual user. It was shown that NOMA has higher sum-rate than OMA; while for the outage probability, the choices of user rates and power coefficients are critical. In \cite{ding1}, the sum-rate superiority of NOMA to OMA was shown for the two-user cluster case and fixed power allocation. The significance of user-pairing based on channel norm difference was also demonstrated. Further, the outage probability of the stronger user was analyzed given quality-of-service (QoS) guarantee of the weak user. In \cite{chen}, for multiple users, the sum-rate optimization over power allocation with fairness consideration were studied for both OMA and NOMA systems and the optimized sum-rate of NOMA was shown to be higher. 

There are also results on multi-input-multi-output (MIMO) NOMA, where the BS is equipped with multiple antennas, e.g., \cite{zeng1,zeng2,senel}. The work in \cite{zeng1} proved that for 2-user case, an upper bound of the sum-rate in OMA systems also serves as a lower bound of the sum-rate in NOMA systems when applying the same precoding and postcoding to OMA and NOMA. The work was generalized to multiple-user case in \cite{zeng2}. In \cite{senel}, for a massive MIMO scenario, the sum-rate achieved by NOMA with two-user clusters was studied and compared to the multi-user beamforming scheme. A hybrid of NOMA and multi-user beamforming scheme was also proposed for sum-rate improvement.

Different from existing works, we study the power consumption and the tradeoff between power consumption and outage probability for MIMO-NOMA systems with a two-user cluster and matched-filter (MF) precoding according to the channel of the stronger user. Considering the significance of channel correlation coefficient (CC) \cite{senel}, we propose a modified correlation-based NOMA (CB-NOMA) scheme, where NOMA is used only when the absolute CC of the user channels is above a threshold. A tight approximation is derived for the average minimal transmit power of CB-NOMA to guarantee QoS for both users. The result reveals the behaviour of the average minimum transmit power with respect to the correlation-threshold and shows the superiority of CB-NOMA\ to the original NOMA in power saving. Further, for the massive MIMO scenario, the scaling laws of the power consumption and outage probability are derived  to guide the threshold design.

\section{NOMA scheme with Correlation Threshold}
We consider the downlink transmissions in a multi-user system with an $M$-antenna BS and 2 single-antenna users. The channel vector for User $k$ can be modeled as:
    ${\bm g}_{k} =\sqrt{\beta_{k}} {\bm h}_{k}$, for $k=1,2,$
where $\beta_{k}$ is the large-scale fading coefficient and $\bm{h}_{k} \in \mathbb{C}^{M \times 1}$ represents the small-scale fading effect with independent and identically distributed (i.i.d.)~elements following the Rayleigh distribution, i.e., $\bm{h}_{k} \sim \mathcal{CN}(\bm{0},\bm{I}_M)$ for $k=1,2$. 
Let $\rho$ be the absolute CC between $\bm{h}_1$ and $\bm{h}_2$, which is given by:
\begin{equation}
    \rho\triangleq\frac{\left| \bm{h}_1^H \bm{h}_2 \right|}{\lVert \bm{h}_1 \rVert \lVert \bm{h}_2 \rVert}.
\end{equation}
In other words, $\rho=|\cos\theta|$, where $\theta$ is the angle between the two channel vectors. We assume perfect CSI at the transmitter.  Without loss of generality, we assume that $\beta_1\geq\beta_2$. 

To save transmit power and achieve a balance between transmit power and the QoS performance, we propose a modified CB-NOMA scheme in which the BS transmits to both users with a common time-frequency resource block as well as common beamformer only when $\rho$, the absolute correlation coefficient, exceeds a threshold $\rho_{th}$. Otherwise the BS keeps silent to save power.
The design of $\rho_{th}$ and its effects on  the performance will be analyzed in later sections. If $\rho_{th}=0$, the scheme is the same as the original NOMA. 

When the BS serves both users, the BS data symbols are superposition-coded as \cite{saito}:
\begin{equation}
    s=\sqrt{P_1}s_{1}+\sqrt{P_2}s_{2},
\end{equation}
where $s_1,s_2\sim \mathcal{CN}(0,1)$ are the data symbols for the users and $P_1,P_2$ are the power allocated to the users, respectively. The total BS transmit power is thus $P=P_1+P_2$.

Let $\bm{b}$ be the BS beamformer. Thus, the transmit vector from the BS is $\bm{b} s$ and the signal received by users is given by:
\begin{equation}
    y_{k}=\bm{b}^H \bm{g}_{k}s+n_{k},~~k=1,2,
\end{equation}
where $n_{k}\sim \mathcal{CN}(0,1)$ is received noise. Since $\beta_1\ge \beta_2$, User 1 has is the stronger user with a higher average channel gain. According to the principle of successive interference cancellation (SIC), User 2 decodes $s_2$ from $y_2$ treating $s_1$ as noise.  User 1 firstly decodes $s_2$ from $y_1$  treating $s_1$ as noise, then cancels the component of $s_2$ from $y_1$, and after that decodes $s_1$ without interference. The signal-to-interference-plus-noise-ratios (SINRs) for User $k$ to decode $s_l$ are thus given by:
\begin{eqnarray}
&&        {\rm SINR}_{1,s_1}=P_{1} \beta_{1} \left| \bm{b}^H \bm{h}_{1} \right|^2,
 \label{SINR_s1} \\
&&        {\rm SINR}_{k,s_2}=\frac{P_{2} \beta_{k} \left| \bm{b}^H \bm{h}_{k} \right|^2}{1+ P_{1} \beta_{1} \left| \bm{b}^H \bm{h}_{k} \right|^2},~k=1,2. \label{SINR_s2}
\end{eqnarray}

It is noteworthy that the decoding order of our CB-NOMA scheme depends on user path-loss coefficients only, equivalently, depends on the \textit{average} channel gains.  This is different from the NOMA schemes in \cite{zeng1} and \cite{zeng2}, where the decoding order depends on the \textit{instantaneous} channel gains, i.e., depends on both $\beta_k$'s and $\bh_k$'s. The latter requires extra communications between the BS and the users about the instantaneous CSI. 

\section{Average Minimal Transmit Power with SINR Guarantee}
In this section, the average minimal transmit power for the CB-NOMA scheme to guarantee QoS of both users are studied. Denote $\gamma_1$ and $\gamma_2$ as the SINR requirements for the users. Since $s_2$ needs to be decoded successfully by both users while $s_1$ only needs to be decoded by User $1$, to guarantee the QoS requirements, we need:
\setlength{\arraycolsep}{1pt}
\begin{eqnarray}
&&\min\left( {\rm SINR}_{1,s_2},{\rm SINR}_{2,s_2} \right) \geq \gamma_2, \label{requirement_gamma_2}\\
\text{and }&&    {\rm SINR}_{1,s_1}\geq \gamma_1 \label{requirement_gamma_1}
\end{eqnarray}

By using $P_1=P-P_2$ and \eqref{SINR_s2} into \eqref{requirement_gamma_2}, we have
\begin{equation}
    P_2 \geq \frac{\gamma_2}{1+\gamma_2} \left( \frac{1}{ \min\left( \beta_{1} \left| \bm{b}^H \bm{h}_{1} \right|^2,\beta_{2} \left| \bm{b}^H \bm{h}_{2} \right|^2 \right)}+P \right). \label{lower_bound_P_2_1} 
\end{equation}

By using \eqref{lower_bound_P_2_1} into \eqref{SINR_s1} and $P_1=P-P_2$, it can be shown that
\begin{equation}
    \begin{aligned}
        {\rm SINR}_{1,s_1}\le&\frac{1}{1+\gamma_2} P  \beta_{1} \left| \bm{b}^H \bm{h}_{1} \right|^2 -\\
        &\frac{ \gamma_2}{1+\gamma_2} \cdot \frac{ \beta_{1} \left| \bm{b}^H \bm{h}_{1} \right|^2}{  \min\left( \beta_{1} \left| \bm{b}^H \bm{h}_{1} \right|^2,\beta_{2} \left| \bm{b}^H \bm{h}_{2} \right|^2 \right)}. \label{SINR_1}
    \end{aligned}
\end{equation}
Finally, substituting \eqref{SINR_1} into \eqref{requirement_gamma_1}, we can obtain
\begin{equation}
    P \geq P_{\rm min} \triangleq\frac{\gamma_1 (1+ \gamma_2) }{ \beta_{1} \left| \bm{b}^H \bm{h}_{1} \right|^2} +\frac{\gamma_2 }{  \min\left( \beta_{1} \left| \bm{b}^H \bm{h}_{1} \right|^2,\beta_{2} \left| \bm{b}^H \bm{h}_{2} \right|^2 \right) }, \label{P_min_Arbitrary_Beamforming}
\end{equation}
with equality if and only if (\ref{lower_bound_P_2_1}) takes equality.
$P_{\rm min}$ defined in (\ref{P_min_Arbitrary_Beamforming}) is the achievable lower bound on the instantaneous transmit power to guarantee the SINR constraints for both users. The condition in \eqref{P_min_Arbitrary_Beamforming} is only a necessary condition since $\gamma_1$ and $\gamma_2$ may not be achieved even when $P \geq P_{\rm min}$ with a non-optimal power allocation. On the other hand, if the optimum power allocation is used, the condition is both necessary and sufficient, that is, $\gamma_1$ and $\gamma_2$ are guaranteed for the two users if and only if $P \ge P_{\rm min}$.

We consider the matched filter (MF) beamforming based on channel vector of the stronger user, i.e., 
\begin{equation}
    \bm{b}=\frac{\bm{h}_1}{\lVert \bm{h}_1 \rVert}. \label{precoder_MRT}
\end{equation}
The minimal transmit power can be rewritten as:
\begin{equation}
    P_{\rm min} =\frac{\gamma_1 (1+ \gamma_2) }{ \beta_1  \lVert \bm{h}_1 \rVert^2} +\frac{\gamma_2 }{ \min \left(\beta_1  \lVert \bm{h}_1 \rVert^2, \beta_2  \lVert \bm{h}_2 \rVert^2 \rho^2 \right) }. \label{P_min_MRT_Beamforming}
\end{equation}
The following theorem is proved for the average value of $P_{\min}$.

\begin{theorem}
Define
\begin{equation}
\tilde{P}_{lo}\triangleq
\frac{ \gamma_1 (1+\gamma_2)}{(M-1)\beta_1}+\frac{ \gamma_2}
{(M-1)\beta_2 \rho_{th}^2} F(1,1;M;1-\rho_{th}^{-2}), 
\label{EP_App_1} 
\end{equation}
where $F(\cdot,\cdot;\cdot;\cdot)$ is the hypergeometric function and $\rho_{th}$ is the CC threshold. The average minimal transmit power for CB-NOMA to guarantee SINR levels of both users, $\gamma_1$ and $\gamma_2$,  has the following lower and upper bounds:
\begin{eqnarray}
\tilde{P}_{lo}\le \Exp \left[P_{\rm min}\right]\le 
\tilde{P}_{lo}
\left(1+\min\left\{\frac{\beta_2}{\beta_1},\frac{\gamma_2}{\gamma_1}\right\} \right).
\label{Results-1}
\end{eqnarray}
\label{thm-1}
\end{theorem}

\begin{proof}
See Appendix \ref{app-A}.
\end{proof}

Theorem \ref{thm-1} provides a lower and an upper bound on the minimal average transit power. It is the foundation of our analysis on the average minimal transmit power. In what follows, we provide several corollaries based on Theorem \ref{thm-1} for more insights.

\begin{corollary} $\lim_{\rho_{th}\rightarrow 0} \Exp[P_{\min}]=\infty$ and $\Exp[P_{\min}]<\infty$ for any $\rho_{th}>0,\gamma_1>0,\gamma_2>0$.
\end{corollary}

This result means that if $\rho_{th}=0$, i.e., no constraint on the CC for using NOMA, any user SINR constraints cannot be guaranteed with finite average transmit power. This is problematic in energy efficiency.  \eqref{LB_1} shows that the unbounded average power is caused by the scenario when $\rho_{th}$ is in the vicinity of zero, i.e., the two channel vectors are close-to-orthogonal.  Naturally, a beam cannot serve the user with channel vector orthogonal to the beam, (which is the weaker user as the beamformer is based on the channel of the stronger user). Thus the SINR constraint of that user $\gamma_2>0$ can never be achieved. This motivates the proposed CB-NOMA scheme with a non-zero threshold on the channel correlation coefficient.       

\begin{corollary} 
When $\gamma_1 \gg \gamma_2$ or $\beta_1 \gg \beta_2$, the average minimal power of CB-NOMA can be tightly approximated as $\tilde{P}_{lo}$, i.e., $\Exp[P_{\min}]\approx \tilde{P}_{lo}$.
\label{cor-2}
\end{corollary}

The result can be obtained directly from (\ref{Results-1}). The application scenario for NOMA is to serve both users at the same time-frequency block when one user has a significantly stronger channel \cite{saito} and a large difference on large scale fading coefficients, i.e.,  $\beta_1\gg \beta_2$ is beneficial to NOMA systems \cite{ding1}. Given the difference in the average channel gains, it is also reasonable to expect a significantly better service to the stronger user, i.e., $\gamma_1 \gg \gamma_2$. For these scenarios, Corollary \ref{cor-2} provides a tight approximation on the average minimal transmit power for any SINR constraints. 

The approximation $\tilde{P}_{lo}$ given in (\ref{EP_App_1}) shows the behaviour of the average minimal transmit  power with respect to the SINR constraints and other network parameters.
For example, $\tilde{P}_{lo}$ increases with $\gamma_1$ and $\gamma_2$ while decreases with $\rho_{th}$, $\beta_1$, $\beta_2$ and $M$. To further explore the behavior of $\Exp[P_{\min}]$ with respect to $M$ and $\rho_{th}$, we introduce the following asymptotic result. 

\begin{corollary} 
When $\gamma_1\gg\gamma_2$ or $\beta_1\gg \beta_2$, for any fixed $M$, when $\rho_{th}\rightarrow 0$, \label{col-3}
\begin{eqnarray} 
&&\hspace{-4mm}\Exp[P_{\min}]\approx \frac{ \gamma_1 (1+\gamma_2)}{(M-1)\beta_1}-
\frac{ \gamma_2}{\beta_2} \frac{\ln{\rho_{th}^2}\hspace{-0.5mm}+ \hspace{-0.5mm}\psi(M\hspace{-0.5mm}-\hspace{-0.5mm}1)
\hspace{-0.5mm}+\hspace{-0.5mm}C \hspace{-1mm}}{\left( 1-\rho^2_{th} \right)^{M-1}}, \label{P-min-small-rho}
\end{eqnarray}where $\psi(\cdot)$ is the di-gamma function and $C\approx 0.5772$ is the Euler-Mascheroni constant.
\label{lemma-1}
\end{corollary} 

\begin{proof}See Appendix \ref{app-B}. \end{proof}

(\ref{P-min-small-rho}) provides a closed-form expression for the average power when $\rho_{th}$ is close to zero. It shows that the average power to guarantee SINR constraints increases as $\ln{(1/\rho_{th}^2})$ for small threshold. 

Next, we consider massive MIMO where $M\gg 1$ and study the asymptotic behaviour with respect to the scaling of both $M$ and $\rho_{th}$ jointly. Both the average transmit power and the outage probability are considered to see the tradeoff. With the proposed CB-NOMA, the users are in outage if and only if the BS is silent, i.e., the absolute CC is smaller than the threshold, the outage probability is thus,
\begin{equation}
    P_{\rm out}=\mathbb{P} \left[ \rho<\rho_{th} \right]=1-\left( 1-\rho_{th}^2 \right)^{M-1}, \label{P_out}
\end{equation}
which is an increase function of $M$ and $\rho_{th}$. 

\eqref{P_out} shows that with a fixed $M$, the outage probability increases as $\rho_{th}$ increases; at the same time, the power consumption decreases. Thus we can adjust the balance between power consumption and outage performance via the design of $\rho_{th}$. Notice that with a fixed $\rho_{th}$ value, the outage probability increases as $M$ increases, though the power  consumption reduces. Thus for massive MIMO systems, a threshold design where $\rho_{th}$ decreases with $M$ is desirable. For this matter, we have the following results.

\begin{corollary}  When $M\rightarrow \infty$ and $\rho_{th}^2=\lambda/M^{\tau}$ for a constant $\lambda>0$, the following results on the average transmit power and the outage probability can be obtained:
\begin{itemize}
\item when $\tau>1$, $P_{\rm out}\rightarrow 0$ and $\tilde{P}_{lo}\rightarrow \infty$;
\item when $\tau<1$, $P_{\rm out}\rightarrow 1$ and $\tilde{P}_{lo}\rightarrow 0$;
\item when $\tau=1$, $P_{\rm out}\rightarrow 1-e^{-\lambda}$ and $\tilde{P}_{lo}\rightarrow \frac{\gamma_2}{\beta_2} e^{\lambda}  E_1(\lambda)$,
\end{itemize}
where $E_1(\cdot)$ is the exponential integral function.
\label{col-4}
\end{corollary} 

\begin{proof}
See Appendix \ref{app-c}.
\end{proof}

Corollary 4 shows the limits of $\Exp[P_{\min}]$ and $P_{out}$. The two performance measures naturally compete with each other since a higher outage probability means less transmissions and less power consumption realized with a higher $\rho_{th}$. The most interesting threshold design is when $\tau=1$, meaning that the square of threshold decreases linearly in $M$, i.e., $\rho_{th}^2 =\lambda /M$ for a fixed $\lambda$. In this case, both $\Exp[P_{\min}]$ and $P_{out}$ have non-trivial bounded limits and by adjusting the value of $\lambda$, we can achieve a continuous tradeoff curve for the power consumption and outage probability. Another observation is that the limits are independent of $\gamma_1$ and $\beta_1$, the two parameters of the stronger user. The outage probability limit is also independent of the parameters  of the weaker user, while the average power consumption depends on $\gamma_2$ and $\beta_2$, meaning that the power consumption of CB-NOMA in a massive MIMO scenario is dominated by the weaker user.

\section{Numerical Result}
In this section, simulation results are demonstrated to show the performance of the proposed CB-NOMA scheme, as well as to verify the the accuracy of our analytical results. 
\begin{figure}
\centerline{\includegraphics[width=\columnwidth]{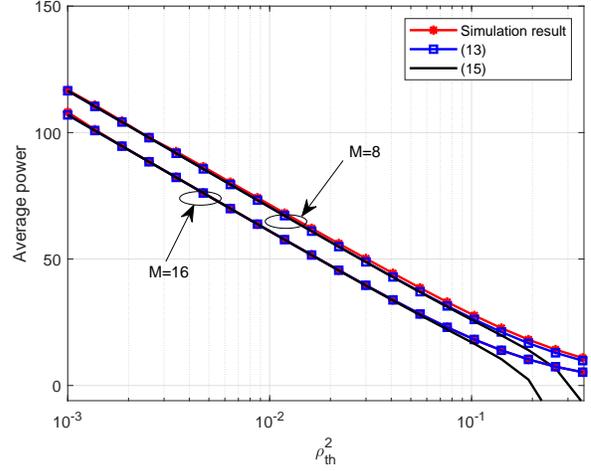}}
\caption{Average minimal power versus $\rho_{th}^2$ where $M=8$ and $16$, $\beta_1=0$dB, $\beta_2=-10$dB, $\gamma_1=10$dB, $\gamma_2=0$dB.}
\label{fig-1}
\end{figure} 
\begin{figure}
\centerline{\includegraphics[width=\columnwidth]{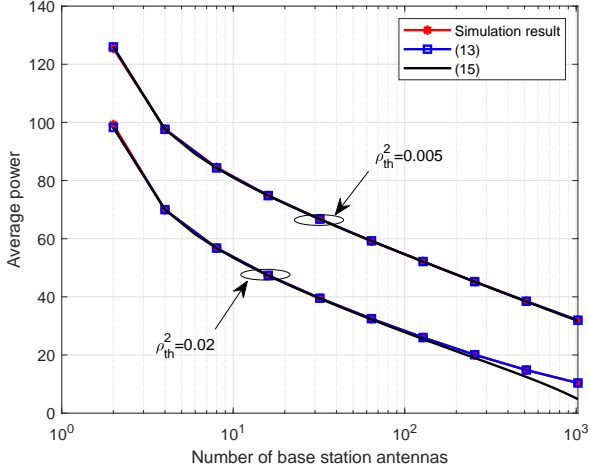}}
\caption{Average minimal power versus $M$ where $\rho_{th}^2=0.02$ and $0.005$, $\beta_1=0$dB, $\beta_2=-10$dB, $\gamma_1=10$dB, $\gamma_2=0$dB.}
\label{fig-2}
\end{figure}
\begin{figure}
\centerline{\includegraphics[width=\columnwidth]{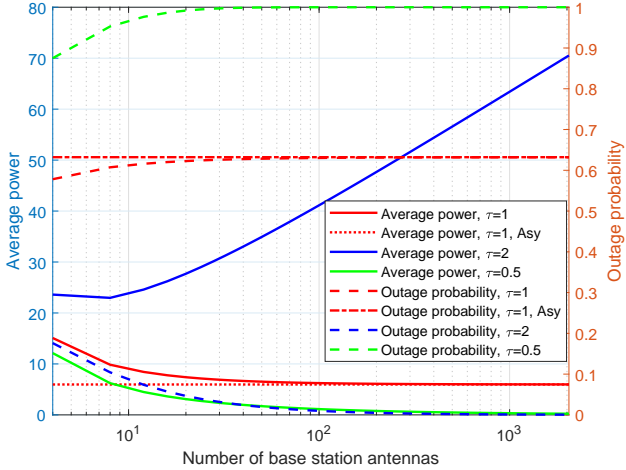}}
\caption{Average minimal power versus $M$ where $\rho_{th}^2=1/M^{\tau}$, $\beta_1=0$dB, $\beta_2=-10$dB, $\gamma_1=10$dB, $\gamma_2=0$dB}
\label{fig-3}
\end{figure}

In Fig.~\ref{fig-1}, the average minimal  power of the CB-NOMA scheme to guarantee the SINR requirements for the two users is shown as a function of $\rho_{th}^2$, where $\beta_1=0$dB, $\beta_2=-10$dB, $\gamma_1=10$dB, $\gamma_2=0$dB. It can be seen that the average minimal power decreases with $\rho_{th}^2$. Further, the figure also shows that $\tilde{P}_{lo}$ is an accurate approximation for all parameter values, while the result in (\ref{P-min-small-rho}) is tight for small $\rho_{th}^2$, e.g., when $\rho_{th}^2<0.1$. We can also see that the average minimal transmit power increases linearly in $\log (1/\rho_{th}^2)$ as $\rho_{th}^2$ approaches 0, which verifies the asymptotic behaviour in Corollary \ref{col-3}.  

Fig.~\ref{fig-2} depicts the average minimal power versus $M$ with the same values of $\gamma_1,\gamma_2,\beta_1,\beta_2$  as before. The accuracy of the approximations in Corollary (\ref{EP_App_1}) and (\ref{P-min-small-rho}) is verified again. The average minimal transmit power decreases as the number of antennas $M$ increases.

Fig.~\ref{fig-3} shows the average minimal power and outage probability versus $M$ where $\rho_{th}^2=1/M^{\tau}$, while other parameter values are the same as before. When $\tau=1$, the average minimal power decreases with $M$ and converges to a positive constant, which matches perfectly with the asymptotic result; when $\tau=0.5<1$, it decreases with $M$ and approaches to $0$; when $\tau=2>1$, it increases unbounded with $M$, which validate the results in Corollary \ref{col-4}. And the behavior of the outage probability also matches the results in Corollary \ref{col-4}.

\section{Conclusion}

A modified NOMA scheme based on the channel correlation coefficient is proposed in this work for systems with multiple-antennas at the BS and two single-antenna users. We derived the average minimal transmit power with QoS guarantee for both users and proved that  a positive channel correlation threshold is required for finite average transmit power. The behaviour of the  average transmit power with respect to the threshold and the BS antenna number is shown. Moreover, to balance the outage probability and average transmit power in massive MIMO systems, we proposed to design the threshold as a decreasing function with the number of BS antennas. The  scaling laws of the outage probability and average minimal power as well as their tradeoff law are derived based on the threshold design.

\begin{appendices}
\section{Proof of Theorem \ref{thm-1}}
\label{app-A}
Since $\bm{h}_1$ and $\bm{h}_2$ are independent each following $\mathcal{CN}(\b0,\bI)$, $\lVert \bm{h}_1 \rVert^2$ and $\lVert \bm{h}_2 \rVert^2$ are independent following the Gamma distribution with shape parameter $M$ and scale parameter 1; $\rho^2$ follows the Beta distribution with parameters 1 and $M-1$ \cite{Jindal}; and $\|\bm{h}_1\|^2,\|\bm{h}_2\|^2,\rho^2$ are mutually independent.

From \eqref{P_min_MRT_Beamforming} we can find the following lower bound of $P_{\rm min}$:
\begin{equation}
    P_{\rm min} \geq P_{lo}=\frac{\gamma_1 (1+ \gamma_2) }{ \beta_1  \lVert \bm{h}_1 \rVert^2} +\frac{\gamma_2 }{ \beta_2  \lVert \bm{h}_2 \rVert^2 \rho^2  }. \label{P_min_MRT_Beamforming_2}
\end{equation}
Thus, by using the probability density functions (PDFs) of Gamma distribution and Beta distribution, the following can be derived:
\begin{align}
    &\mathbb{E}\left[ P_{\rm min} \right]\geq\iiint_V f_{\lVert \bm{h}_1 \rVert^2}(x)    f_{\lVert \bm{h}_2 \rVert^2}(y)   f_{\rho^2}(z ) P_{{\rm min,lo}} {\rm d} x {\rm d} y{\rm d} z \notag\\
    =& \underbrace{\frac{ \gamma_1 (1+\gamma_2)}{(M-1)\beta_1}}_{T_0}+\underbrace{\iiint_{V}  f_{\lVert \bm{h}_2 \rVert^2}(y)   f_{\rho^2}(z ) \frac{\gamma_2 }{ \beta_2  y z  } {\rm d} x {\rm d} y{\rm d} z}_{T} \label{LB_0}\\
    =&T_0+\frac{ \gamma_2}{\left( 1-\rho^2_{th} \right)^{M-1}  \beta_2} \int_{\rho^2_{th}}^1 \frac{\left( 1-\rho^2 \right)^{M-2}}{\rho^2}{\rm d} \rho^2=\tilde{P}_{lo},\label{LB_1}
\end{align}
where $V=\left\{ (x, y, z)| x\in (0,\infty), y\in(0,\infty),z\in\left[\rho_{th}^2,1\right] \right\}$,  $f_X(\cdot)$ represents the PDF, and the last step can be obtained from the definition of hypergeometric function with some mathematical manipulations.

Next, we show the upper bound. First, define $V_1\triangleq\{ (x, y, z)| \beta_1 x < \beta_2 y z \}$ and $V_2\triangleq\{ (x, y, z)|\beta_1 x \geq \beta_2 y z \}$. By noticing that $V=V_1 \cup V_2$, we have from \eqref{LB_0},
\begin{equation}
    \begin{aligned}
        \mathbb{E}\left[ P_{\rm min} \right]
        =&T_0+\underbrace{\iiint_{V_1}f_{\lVert \bm{h}_1 \rVert^2}(x) \frac{\gamma_2 }{ \beta_1  x }{\rm d} x {\rm d} y{\rm d} z}_{T_1}+\\
        &\underbrace{\iiint_{V_2}  f_{\lVert \bm{h}_2 \rVert^2}(y)   f_{\rho^2}(z ) \frac{\gamma_2 }{ \beta_2  y z  } {\rm d} x {\rm d} y{\rm d} z}_{T_2}.\label{average_P_min_2}
    \end{aligned}
\end{equation}
Since $V_2 \subseteq V$, we have $T_2\le T$ and thus:
\begin{equation}
    T_0+T_2 \leq T_0+T=\tilde{P}_{lo}. \label{Ineq_1}
\end{equation}
For $T_1$, we have
\setlength{\arraycolsep}{1pt}
\begin{eqnarray}
        T_1&\leq & \frac{\gamma_2}{\gamma_1} \iiint_{V}f_{\lVert \bm{h}_1 \rVert^2}(x) \frac{\gamma_1 }{ \beta_1  x }{\rm d} x {\rm d} y{\rm d} z \nonumber\\
        \leq &&\hspace{-3mm} \frac{\gamma_2}{\gamma_1} \iiint_{V}f_{\lVert \bm{h}_1 \rVert^2}(x) \frac{\gamma_1 (1+\gamma_2)}{ \beta_1  x }{\rm d} x {\rm d} y{\rm d} z=\frac{\gamma_2}{ \gamma_1 } T_0, \label{UB_1}
\end{eqnarray}
and
\begin{eqnarray}
        T_1&\leq &  \frac{\beta_2}{\beta_1} \iiint_{V}f_{\lVert \bm{h}_2 \rVert^2}(y)  \frac{\gamma_2 }{ \beta_2  y  }{\rm d} x {\rm d} y{\rm d} z \nonumber \\
        \leq &&\hspace{-3mm}\frac{\beta_2}{\beta_1}\iiint_{V}  f_{\lVert \bm{h}_2 \rVert^2}(y)   f_{\rho^2}(z ) \frac{\gamma_2 }{ \beta_2  y z  } {\rm d} x {\rm d} y{\rm d} z=\frac{\beta_2}{\beta_1} T. \label{UB_2}
\end{eqnarray}
From \eqref{UB_1} and \eqref{UB_2} we can obtain:
\begin{equation}
    T_1 \leq \min\left\{ \frac{\beta_2}{\beta_1},\frac{\gamma_2}{\gamma_1} \right\} \left( T_0+T \right)=\min\left\{ \frac{\beta_2}{\beta_1},\frac{\gamma_2}{\gamma_1} \right\} \tilde{P}_{lo}. \label{Ineq_2}
\end{equation}

By combining \eqref{average_P_min_2}, \eqref{Ineq_1} and \eqref{Ineq_2}, the upper bound of $\mathbb{E} \left[ P_{\rm min} \right]$ in \eqref{Results-1} is proved.

\section{Proof of Corollary \ref{lemma-1}}
\label{app-B}
When $\gamma_1\gg\gamma_2$ or $\beta_1\gg \beta_2$, we have from Corollary 2 $\Exp[P_{\min}]\approx \tilde{P}_{lo}$ as defined in \eqref{EP_App_1}. In Appendix A, an alternative form of $\tilde{P}_{lo}$ is derived in \eqref{LB_1}. The integral in \eqref{LB_1} can be further calculated as follows: 
\begin{equation*}
\begin{aligned}
&\int_{\rho_{th}^2}^1 \frac{\left( 1-\rho^2 \right)^{M-2}}{\rho^2}{\rm d}\rho^2\\ = & \sum_{k=0}^{M-2}\left( -1 \right)^{k}\int_{\rho_{th}^2}^1 \dbinom{M-2}{k} \rho^{2(k-1)}d\rho^2\\
=& \sum_{k=1}^{M-2}\frac{\left( -1 \right)^{k}}{k}\dbinom{M-2}{k} \left(1-\rho_{th}^{2k}\right)
-\ln{\rho_{th}^2}\\
=&-\psi(M-1)-C -\ln{\rho_{th}^2}+\mathcal{O}(\rho_{th}^2). 
\end{aligned}
\end{equation*}
When $\rho_{th}^2 \to 0$, by ignoring the higher order terms of $\rho^2_{th}$, \eqref{P-min-small-rho} is obtained. 

\section{Proof of Corollary \ref{col-4}}
\label{app-c}

The limits of $P_{out}$ in Corollary \ref{col-4} can be obtained by:
\begin{equation}
    \begin{aligned}
        \lim_{M \to \infty} \left( 1-\frac{\lambda}{M^\tau} \right)^{M-1}
        =&
        \begin{cases}
        0 &,\tau<1\\
        e^{-\lambda}~~&,\tau=1\\
        1 &,\tau>1.
        \end{cases} \label{P_limit}
    \end{aligned}
\end{equation}

For the average transmit power, we consider the alternative form of $\tilde{P}_{lo}$ given by \eqref{LB_1}. When $\rho_{th}^2=\frac{\lambda}{M^\tau}$, the limit of $\tilde{P}_{lo}$ is given by:
\begin{equation}
    \lim_{M \to \infty} \tilde{P}_{lo}= \lim_{M \to \infty} \left[  \frac{ \gamma_2}{\left( 1-\frac{\lambda}{M^\tau} \right)^{M-1}  \beta_2} I \right], \label{P_lo_limit}
\end{equation}
where
\begin{equation}
\begin{aligned}
    I\triangleq\int_{\frac{\lambda}{M^\tau}}^1 \frac{\left( 1-\rho^2 \right)^{M}}{\rho^2}{\rm d} \rho^2\xlongequal[]{y=M \rho^2} &\int_{\lambda M^{1-\tau}}^M \frac{\left( 1-\frac{y}{M} \right)^{M}}{y}{\rm d} y.
    \label{Integral_with_threshold}
\end{aligned}
\end{equation}

Now notice that $(1+\frac{y}{M})^M<e^y<(1+\frac{y}{M})^{M+1}$ for any $M\ge 1$, thus
\begin{equation}
\begin{aligned}
    1>
   (1-\frac{y^2}{M^2})^M=&(1-\frac{y}{M})^M[(1+\frac{y}{M})^{M+1}]^{\frac{M}{M+1}}\\
    >& 
    \left(1-\frac{y}{M}\right)^M e^{\frac{M}{M+1}y},
\end{aligned}
\end{equation}
\begin{equation}
    (1-\frac{y^2}{M^2})^M =(1-\frac{y}{M})^M(1+\frac{y}{M})^{M}    
    < 
    (1-\frac{y}{M})^M e^y.
\end{equation}
These gives, for all $M\geqslant 1$, 
\begin{equation}{\label{eq:202001072152}}
    \frac{(1-\frac{y^2}{M^2})^M}{y}e^{-y}
    \leqslant
    \frac{(1-\frac{y}{M})^M}{y}
    \leqslant
    \frac{1}{y}e^{-\frac{y}{2}}.
\end{equation}
Consequently, when $\tau<1$, we have, as $M\to\infty$, 
\begin{equation}
    I\leqslant \int_{\lambda M^{1-r}}^M 
    \frac{1}{y} e^{-\frac{y}{2}}
    {\rm d}y
    \leqslant 
    \frac{1}{\lambda} M^{r-1} 
    \int_0^\infty e^{-\frac{y}{2}}
    {\rm d}y
    \longrightarrow 0. 
\end{equation}
When $\tau>1$, 
\begin{equation}
    I\geqslant \int_{\lambda M^{1-r}}^1 
    \frac{(1-\frac{1}{M^2})^M}{y}e^{-y} 
    {\rm d}y
    \longrightarrow \infty. 
\end{equation}
Finally when $\tau=1$, by (\ref{eq:202001072152}), and $(1-\frac{y}{M})^M\longrightarrow e^{-y}$ for every $y$, we have 
\begin{equation}
    I\longrightarrow \int_\lambda^\infty 
    \frac{e^{-y}}{y}{\rm d}y
\end{equation}
through Lebesgue's dominated convergence theorem. Thus ends the proof.

\end{appendices}

\end{document}